\documentclass[runningheads]{llncs}

\usepackage{ stmaryrd }
\usepackage{amsmath}

\usepackage{xspace,amssymb,amsmath}
\usepackage{xcolor}
\usepackage{hyperref}
\hypersetup{
  pdftitle = {},
  pdfauthor = {},
  colorlinks = true,
  linkcolor = black!30!red,
  citecolor = black!30!green
}

\usepackage{tikz}

\usepackage{indentfirst}

\binoppenalty=\maxdimen
\relpenalty=\maxdimen

\newcommand{\quotes}[1]{``#1''}

\DeclareMathOperator{\score}{score}
\DeclareMathOperator{\id}{id}

\title{Lower bound for constant-size local certification\thanks{J.-F. Raymond was
supported by the ANR project GRALMECO (ANR-21-CE48-0004-01). L. Feuilloley was supported by ANR project GrR (ANR-18-CE40-0032).}}

\institute{
LAMSADE, Université Paris Dauphine, France
\and Laboratoire G-SCOP, Univ. Grenoble Alpes
\and CNRS, LIRIS, Univ. Lyon, France
\and LaBRI, Université de Bordeaux, CNRS, Bordeaux, France
\and ENS Lyon, LIP, Lyon, France
\and CNRS, LIMOS, Université Clermont Auvergne, France
}

\author{Virgina Ardévol Martínez\inst{1} \and Marco Caoduro\inst{2} \and Laurent Feuilloley\inst{3}\orcidID{0000-0002-3994-0898}\and Jonathan Narboni\inst{4}\orcidID{0000-0002-3087-5073} \and Pegah Pournajafi\inst{5} \and Jean-Florent Raymond\inst{6}\orcidID{0000-0003-4646-7602}}
\authorrunning{Ardévol Martínez, Caoduro, Feuilloley, Narboni, Pournajafi,  Raymond}

\date{}

\begin{document}
\maketitle

\begin{abstract}
Given a network property or a data structure, a local certification is a labeling that allows to efficiently check that the property is satisfied, or that the structure is correct. 
The quality of a certification is measured by the size of its labels: the smaller, the better.
This notion plays a central role in self-stabilization, because the size of the certification is a lower bound (and often an upper bound) on the memory needed for silent self-stabilizing construction of distributed data structures. From the point of view of distributed computing in general, it is also a measure of the locality of a property (e.g. properties of the network itself, such as planarity).

When it comes to the size of the certification labels, one can identify three important regimes: the properties for which the optimal size is polynomial in the number of vertices of the graph, the ones that require only polylogarithmic size, and the ones that can be certified with a constant number of bits. The first two regimes are well studied, with several upper and lower bounds, specific techniques, and active research questions. On the other hand, the constant regime has never been really explored, at least on the lower bound side. 

The main contribution of this paper is the first non-trivial lower bound for this low regime. More precisely, we show that by using certification on just one bit (a binary certification), one cannot certify $k$-colorability for $k\geq 3$. 
To do so, we develop a new technique, based on the notion of score, and both local symmetry arguments and a global parity argument. 
We hope that this technique will be useful for establishing stronger results. 

We complement this result by a discussion of the implication of lower bounds for this constant-size regime, and with an upper bound for a related problem, illustrating that in some cases one can do better than the natural upper bound. 
\end{abstract}

\newpage{}

\section{Introduction}

Local certification consists in assigning labels to the nodes of a network, to allow them to check locally that some property holds~\cite{Feuilloley21}. 
Historically, the concept appeared implicitly in the study of self-stabilization, where in addition to computing the solution of the problem, the nodes would compute some additional information that would allow fast checking of the solution (i.e. a certification of the solution).
The most classic example is maybe the problem of computing a spanning tree, where in addition to computing the pointer to its parents, every node stores its distance to the root~\cite{AfekKY90}.
Such additional information is an overhead in the memory used, hence, it is a natural goal to minimize its size. 
In \cite{BlinFP14}, Blin, Fraigniaud, and Patt-Shamir proved that for a standard notion of self-stabilization called \emph{silent stabilization} and up to some hypothesis, the space needed for certification is the same as the space needed for self-stabilization

More recently, the question of the certification of properties of the network itself, and not the correctness of a data structure built on top of it, has attracted a lot of attention. This paper follows this direction, and we will only consider certification of graph properties. 

Let us now give a more precise intuition of what a local certification is (proper definitions will be given in Section~\ref{sec:model}). We denote the number of vertices of a graph by $n$. 
For a graph property $P$, we will say that it has a local certification of size $s$ if:
\begin{itemize}
    \item for any graph $G$ such that $P$ holds, there exists an assignment of labels of size $s(n)$ that can \quotes{convince} all the nodes that $P$ is satisfied,
    \item for any graph $G$ such that $P$ \emph{does not hold}, for any assignment of labels of size $s(n)$, there is at least one node that detects that the property is not satisfied.
\end{itemize}

At the level of the nodes, the behavior is the following ; every node runs the same local decision algorithm that takes as input all the information available in a neighborhood (i.e. a local view), and outputs a decision: \emph{accept} or \emph{reject}. 
For positive instances, all the nodes are convinced, that is, they all accept. For negative instances, at least one node rejects.

There are actually many possible models, depending on the notion of neighborhood considered, the presence of identifiers and how nodes can use them, etc. 
Two classic models are proof-labeling schemes~\cite{KormanKP10} and locally checkable proofs~\cite{GoosS16}.
The precise model is not essential for the discussion that follows, hence we delay their definitions to the model section.

\subsection{Three typical regimes for the certificate sizes}

As said earlier, a natural goal in the study of local certification is to minimize the size of the certificates. It is well-known that the optimal size is always in $O(n^2)$, since one can always use the adjacency matrix as a certificate and make the node check the consistency of this matrix with their neighborhoods, as well as check that the property holds in the graph described by the matrix~\cite[Theorem 3.2]{KormanKP10}.

Also, for any subquadratic function $f$, it is possible to engineer a property for which the optimal size is $f(n)$~\cite[Corollary 2.4]{KormanKP10}. In other words, if we consider the certificate size as the complexity of a property, there is no gap in the complexity of certification.
Nevertheless, for all the natural properties that have been studied, the optimal certificate size only belongs to one of the following three regimes (already identified in \cite{GoosS16}): polynomial, (poly)logarithmic, and constant size.
For example, there is no known natural properties with certificate size $\Theta(\log \log n)$, or $\Theta( \log^*n)$, or $\Theta(2^{\sqrt{\log n}})$.

In this paper, we are interested in lower bounds for the constant size regime. But let us provide a quick overview of the three regimes, in order to give the full picture and later discuss the novelty of our techniques.  

\paragraph{Polynomial regime.}

It is known that the $\Theta(n^2)$ size is needed for some specific properties, such as having a non-trivial automorphism~\cite[Theorem 6.1]{GoosS16} or having chromatic number at least 4 (up to polylog factors)~\cite[Theorem 6.4]{GoosS16}. 
Even innocent-looking properties such as having diameter 3 or being triangle-free require certificates of size $\Omega(n)$~\cite[Theorem 1]{Censor-HillelPP20} and~$n/e^{O(\sqrt{n})}$~\cite[Proposition 5]{CrescenziFP19}, respectively.

\paragraph{Polylogarithmic regime.}

The regime of (poly)logarithmic certificate size has attracted a lot of attention recently, and is often referred to as \emph{compact certification} (or LogLCP in~\cite{GoosS16}). 
The best-known local certification is the certification of acyclicity (that is, the vertification of the class of trees) for which the optimal size is $\Theta(\log n)$ (by a straightforward adaptation of \cite[Lemma 2.2]{KormanKP10}). 
It has been proved recently that planarity and bounded-genus~\cite{FeuilloleyFMRRT21,FeuilloleyFMRRT21-b,EsperetL22} have logarithmic certification, and that MSO properties on graphs of bounded treedepth~\cite{BousquetFP22} and bounded treewidth~\cite{FraigniaudMRT22} have respectively $\Theta(\log n)$ and $O(\log^2 n)$ local certifications.
An important open question in the area is to establish whether any graph class defined by a set of forbidden minors has a compact certification. Partial results are known for small minors~\cite{BousquetFP21} or minors with specific shapes (namely paths \cite[Corollary 2.7]{BousquetFP22} and planar graphs \cite[Corollary~3]{FraigniaudMRT22}).
Finally, let us mention one key result of the area, even if it is concerned with a data structure instead of a graph property: the optimal certificate size for a minimum spanning tree is $O(\log n \log W)$, where $W$ is the maximum weight~\cite{KormanK07}.

\paragraph{Constant size regime.}

First, let us note that for some properties, no certificate is needed. For example, checking that the graph is a cycle can be done by simply having every node check that it has exactly two neighbors (we will always assume that the graph is connected, thus there must be only one cycle). 

Now, let us make a connection with a class of \emph{construction problems}.
A key class of problems in distributed computing is the construction of locally checkable labelings~\cite{NaorS95}, or LCLs for short, that are the constant-size labelings that can be checked by inspecting a local neighborhood. Examples of LCLs are maximal independent sets, maximal matchings, and minimal dominating sets at some distance~$d$. There is now a very large literature on computing such labelings. 
The problem that has attracted the most interest is the one of vertex coloring: given an integer $k$ (that typically depends on the maximum degree), how fast can we compute an assignment of colors to the vertices, such that for every edge the endpoints have different colors. We refer to the monograph~\cite{BarenboimE13} for distributed graph coloring.

For any LCL, we can design a certification question:

\begin{question}\label{question:LCL}
How many bits are needed to certify that the graph has a solution?
\end{question}

Most classic LCL problems are designed so that any graph has a solution; for example, any graph has a maximal matching. But it is not true for any LCL; for example, given a positive integer $k\in \mathbb{N}$, not all graphs are $k$-colorable. 

When a solution exists, and can be checked by inspecting each node and its direct neighbors, it is trivial to design a certification for the question above: get a solution, give every node its label as a certificate, and for the verification, let the vertices run the local checking. 
Specifically, to certify that a graph is $k$-colorable, one can find a proper $k$-coloring, and then give to every node its color. It is then easy for the vertices to check this certification: every node checks that no neighbor has been given the same color as itself. This certification uses certificates of $\lceil \log k \rceil$ bits, and the key question that we would like to answer is:

\begin{question}\label{question:logk}
Can we do better than $\lceil \log k \rceil$ bits to certify that a graph is $k$-colorable?
\end{question}

This question was already listed in \cite{Feuilloley21} as one of the key open questions in the field, and in the following section, we will review a few reasons why it is a question worth studying.

\subsection{Motivation for studying the constant-size regime}

In local certification, and more generally in theoretical computer science, the focus is usually not on the precise constants in the complexities, thus one might consider the questions above to be non-essential. 
Let us list a few reasons why Questions~\ref{question:LCL} and~\ref{question:logk} are actually important.

\paragraph{An arena for new lower bound techniques.}

The lower bound techniques that we have for local certification are mainly of two types (see the survey~\cite{Feuilloley21} for precise citations and more detailed sketches).

First, there are the techniques based on counting arguments, also called cut-and-plug techniques, that can be rather sophisticated but boils down to the following fact:
if we use $o(\log n)$-bit labels in \emph{yes}-instances (that is, correct instances) then some (set of) labels will appear in different places of an instance (or in different instances), because there are $n$ vertices and $o(n)$ different labels. Using this, one can build a \emph{no}-instance and derive a contradiction. This is the technique used to show almost all lower bounds in the logarithmic regime.

The second classic technique is a reduction from communication complexity, which we will not sketch here, but simply mention that it works better for the polynomial regime. 

At this point, two things are clear: (1) we have only two main techniques, and they are now very well understood, and (2) they do not solve all our problems. 
In particular, they do not seem to apply to the $o(\log n)$ regime. 
One can hope that by trying to give a negative answer to Question~\ref{question:logk}, we will create new techniques, and that these techniques could be useful to establish new lower bounds.

\paragraph{A point of view on the encoding of LCLs.}

As mentioned earlier, the study of LCLs now plays a key role in our understanding of locality in distributed computing. By asking Question~\ref{question:LCL}, we are basically asking about how we express such problems. If it is possible to use fewer bits than the obvious encoding to certify that a solution exists, what does this tell us? If we cannot ``compress'' the encoding, can we say that it is the core of the problem? Such consideration could have an impact on techniques that heavily rely on precise encodings, such as round elimination~\cite{Suomela20} and local conflict colorings~\cite{MausT20,FraigniaudHK16}.

\paragraph{Beyond the constant regime.}

Up to now, we have considered the question of $k$-coloring with constant $k$, thus Question~\ref{question:logk} was about the constant regime \emph{per se}. 
But actually, one could let $k$ depend on $n$, and the question is meaningful for labels of size up to $\Theta(\log n)$ (which corresponds to coloring with $\Theta(n)$ colors). 
Hence, we are not only playing with constants when studying Question~\ref{question:logk} in the general case. Note that if we could show that $k$-coloring requires labels of size $\Theta(\log k)$ for all the range of $k$, then we would also have fairly natural problems strictly between constant and logarithmic, which would be new and interesting in itself.

\paragraph{A candidate for disproving the trade-off conjecture.}

One of the remaining important open questions in local certification is the following, which we will call the \emph{trade-off conjecture}.

\begin{question}
Suppose that there exists a local certification with labels of size $f(n)$ for some property, where every node would check its radius at distance 1. 
Is it true that there always exists a certification with labels of size $O(f(n)/t)$ if we allow the nodes to see at distance $t$ in the graph?
\end{question}

This question and variants of it were raised in~\cite{OstrovskyPR17,FeuilloleyFHPP21,FischerOS21}. In these papers, the authors basically prove  (among other results) that the answer is positive if the certification uses only spanning trees and a uniform certification (giving the same information to every node).
Since these are the main tools used for certification in the logarithmic and polynomial regimes, the constant regime (and its extension beyond constant, discussed above) seems to be the place to find potential counterexamples. Note that for the constant regime, one could try to disprove the conjecture with $\alpha f(n)/t$, for some given constant $\alpha$. And it seems like a reasonable approach, since it is difficult to imagine how the trade-off conjecture could be true for coloring-like problems: even if the nodes can see further, how can we save bits to certify colorability? In order to prove such a counterexample, we first need to have proper lower bounds for distance 1, that is, to answer Question~\ref{question:logk}.

\subsection{Our results and techniques}

In this paper, we give the first non-trivial lower bound  for the certification of $k$-colorability. 
This is the first step of a research direction that we hope to be successful, and in terms of result it is a small step, in the sense that our answer to Question~\ref{question:logk} is restricted in several ways. 
The first restriction we make is that instead of proving a $\log k$ lower bound, we prove that it is not possible to certify $k$-colorability for any $k\geq 3$, if we use \emph{only one bit} (that is, only two different labels). We will call this a \emph{binary certification}.  
The second restriction is that we take a model that is not the most powerful one. 
In our result, a vertex has access to the following information: its identifier, its label, and the multiset of a of its neighbors. That is, a node cannot see further than its direct neighbors and cannot access the identifiers of its neighbors (which corresponds to the original proof-labeling scheme model~\cite{KormanKP10}, but not to the generalizations, such as locally checkable proofs~\cite{GoosS16}), and there is no port number. 

To prove this result, we introduce a new technique. Similarly to the cut-and-plug technique mentioned earlier, we reason about one or several \emph{yes}-instances and prove that we can craft a \emph{no}-instance where the vertices would accept. 
But the reasoning is different, since counting arguments based on the pigeon-hole principle applied to the certificates can only lead to $\Omega(\log n)$ lower bounds. 
First, we define the notion of \emph{score} for a neighborhood and prove that if two vertices have been given different labels but have the same score, then we can build a \emph{no}-instance that is accepted. 
Then we prove that this necessarily happens in some well-chosen graphs, thanks to a series of local symmetry arguments, and a global parity argument.

We complement this main result by proving that in some cases (namely distance-2 3-colorability) one can actually go below the size of the natural encoding. As we will see, this happens because graphs that are distance-2 3-colorable have a very specific shape. This illustrates why establishing lower bounds for such problems is not so easy: the fact that the graph can (or cannot) be colored with a given number of colors implies that it has a given structure, and this structure could in theory be used in the certification to compress the natural $\log k$-bit encoding. 

\section{Models and definitions}
\label{sec:model}

We denote by $ \mathbb N $ the set of non-negative integers, and by $ |A| $ the cardinality of a set $ A $. All graphs in this paper are simple and connected.
The vertex-set and the edge-set of a graph $ G $ are denoted by $V(G)$ and $E(G)$, respectively. The \emph{closed neighborhood} of a vertex $ v \in V(G) $, denoted by $ N[v]$, is defined by $ N[v] = N(v) \cup \{v\}$ where $ N(v) =  \{ u \in V(G): uv \in E(G) \} $ is the neighborhood of $ v $. We denote the complete graph on $ n $ vertices by $ K_n $. 
A \emph{proper $k$-coloring} of the vertex-set of a graph $ G $ is a function $ \phi: V(G) \rightarrow \{1,2,\dots, k\} $ such that if $ xy \in E(G)$, then $ \phi(x) \neq \phi(y)$. In other words, it is an assignment of colors to the vertices of $ G $ using at most $ k$ colors, such that the endpoint of every edge receive different colors. We say that $ G $ is \emph{$k$-colorable} if it admits a proper $ k$-coloring.

Let $ f: \mathbb N \rightarrow \mathbb N \cup \{\infty\} $ be a function. We say that a graph $G$ on $n$ vertices is equipped with an \emph{identifier assignment} of range $f(n)$ if every vertex is given an integer in $[1,f(n)]$ (its \emph{identifier}, or \emph{ID} for short) such that no two vertices of the graph are given the same number. 
Typically, $f(n)$ is some polynomial of $n$, and in this paper it has to be at least $n^3$ (but we did not try to optimize this parameter). 

A \emph{certificate assignment} of size $s$ of a graph $G$ is a labeling of the vertices of $G$ with strings of length $s$, that is, a function $\ell \colon V(G) \to \{0,1\}^s$. A \emph{binary certificate assignment} is a certificate assignment with $s=1$.

As hinted earlier, there are many variants for the definition of local certification. An important aspect is the type of algorithm that the node run.
This is a local algorithm, in the sense that the nodes can see only a neighborhood in the graph, but this neighborhood can be at distance 1, constant, non-constant etc.
Another important aspect is the symmetry-breaking hypothesis: whether there are identifiers, whether the nodes can see the identifiers of their neighbors, whether they can distinguish these neighbors, etc.
In this paper, we use the following notion.

\begin{definition}

A \emph{local decision algorithm} is an algorithm that runs on every vertex of a graph. It takes as input the identifier of the node, the certificate of the node, and the multiset of certificates of its neighbors, and outputs a decision, \emph{accept} or \emph{reject}.
\end{definition}

\begin{definition}

Fix a function $ f: \mathbb N \rightarrow \mathbb N \cup \{\infty\} $. A \emph{proof-labeling scheme of size $s$ for a property $P$} is a local decision algorithm $A$ such that the following holds: for every graph $G$ and every identifier assignment of range $ f(|V(G)|) $ of $ G $ there exists a certificate assignment of size $s$ of $ G $ such that $A$ accepts on every vertex in $ V(G)$, if and only if, the graph $G$ has property $P$.
\end{definition}

Notice that the proof-labeling scheme depends on the chosen function $ f $. 

In the proofs, as a first step, we will prove the result in a weaker anonymous model.

\begin{definition}
An \emph{anonymous proof-labeling scheme} is the same as a proof-labeling scheme, but the graphs are not equipped with identifiers (or equivalently, the outcome of the local decision algorithm is invariant by a change of the identifiers).
\end{definition}

For a graph $ G $ with identifier $ \id $ and labeling $ \ell $, the \emph{view} of a vertex $ v $ in $(G, \id, \ell)$ in the proof-labeling scheme is a tuple $(M_v, \id(v))$ where $M_v$ is the multiset $$\{ (\ell(u), \id(u)) : u \in N(v)\}.$$
In the anonymous case, the \emph{view} of a vertex $ v $ is only the multiset defined above.

\section{$k$-colorability does not have a binary certification}

This section contains our main contribution. We prove that $k$-colorability does not have a binary certification when $k\geq 3$. Recall that for $ k = 2 $, $ k $-colorabily indeed has a binary certification (take the colors as certificates).  

\subsection{Indistinguishability setting}

Let us first clarify the proof strategy with Lemma~\ref{lem:techniques}. It is a classic strategy, that we detail for completeness.

\begin{lemma} \label{lem:techniques} 
Let $ s $ be a positive integer, $ f: \mathbb N \rightarrow \mathbb N \cup \{\infty\}$ be a function, and $ \Lambda \subseteq \mathbb N $ be a set of indices. If for every $ i \in \Lambda $ there exist a connected graph $ G_i$ with identifier $ \id_i : V(G_i) \rightarrow [1, f(|V(G_i)|)]$, and there exists a connected graph $ H $ such that 
\begin{enumerate}
    \item $ G_i $ is $ k$-colorable for every $ i \in \Lambda$,
    \item $ H $ is not $ k$-colorable,
    \item for every set of labelings $ \{\ell_i\}_{i \in \Lambda} $, where $ \ell_i:V(G_i) \rightarrow \{0,1\}^s $ is a labeling of size $ s $ of $ G_i $, there exists a labeling $\ell:V(H) \rightarrow \{0,1\}^s$ of size $ s $ of $ H $ and an identifier $ \id:V(H)~\rightarrow~[1, f(|V(H)|)] $ such that for every view in $ (H,\id, \ell) $ there exists $ i \in \Lambda $ such that the view is the same as some view in $ (G_i, \id_i, \ell_i) $,
\end{enumerate}
then $k$-colorability cannot be certified by certificates of size $ s $. 

In case $ G $ and $ H $ do not have identifiers, the same holds with removing the identifier functions and $ f $ from the statement of the Lemma. 
\end{lemma}

\begin{proof} 
Suppose there exists a local certification of size $s$ for $k$-colorability. Then, in particular, for every $ i \in \Lambda $, there exists a labeling $\ell_i$ for the graph $G_i$ such that the verifier algorithm accepts on every vertex. For this set of labelings $\{\ell_i\}_{i \in \Lambda}$, consider the labeling $\ell$ and the identifier assignment $\id$ of $H$ described in item (3) of the Lemma. The verifier algorithm accepts on every vertex of $ (H, \id, \ell) $ because its view is the same as a view in $ (G_i, \id_i, \ell_i)$ for some $ i \in \Lambda$. This contradicts the fact that $H$ is not $k$-colorable.  \qed
\end{proof}

\subsection{Notion of score}

Let $\ell $ be a binary labeling of a graph $ G $, and let $ v \in V(G)$. The \emph{score} of~$v$ in~$\ell$, denoted by $\score_{\ell}(v)$ or $\score(v)$ if there is no confusion, is defined as follows:
$$
\score_\ell(v) = |\{ u \in N[v] : \ell(u)=1 \}|.
$$

Given a $k$-regular graph $ G $ and a binary labeling $\ell $ of $ G $, the \emph{score matrix} of~$(G,\ell)$ is a $2\times(k+2)$ matrix $S$ with rows labeled with $0$ and $1$ and columns labeled from $0$ to $k+1$. Let $ S_{i,j} $ denote the $(i,j)$ element of $ S $. We set
$$
S_{1,0} = S_{0,k+1} = 1,
$$
and for $ i=0,1$, $ j=0,1,\dots, k+1$, and $(i,j) \neq (1,0), (0,k+1)$ we set
$$S_{i,j} = |\{ v \in V(G) : \ell(v) = i, \score(v) = j \}|.$$

\subsection{Our graph construction and its properties}

Fix an integer $ k\geq 3$. We build a graph as follows: take the disjoint union of three copies of $K_{k+1}$. For $ i = 1,2,3 $, let $a_i $ and $ b_i$ be two distinct vertices in the $i$-th copy. Then, remove the edges $a_1b_1$, $a_2b_2$, and $a_3b_3$ from the graph, and add the edges $b_1a_2$, $b_2a_3$, and $b_3a_1$ to it. We denote the resulting graph by $ N_k $. See Figure~\ref{fig:graph-N_k}. In the figure, each set $C_t$, $t \in \{1,2,3\}$ induces a $K_{k-1}$ in the graph and $a_t$ and $ b_t $ are complete to $C_t$, i.e.\ every vertex of $ C_t$ is connected by an edge to $ a_t$ and to $b_t$.

\begin{figure}[h]
    \centering
    \begin{tikzpicture}[scale=.5]
\draw (70:5cm) -- (110:5cm);    
\draw (190:5cm) -- (230:5cm);
\draw (310:5cm) -- (350:5cm);
\draw (70:5cm) -- (30:5.3cm);
\draw (70:5cm) -- (30:2.7cm);
\draw (110:5cm) -- (150:5.3cm);
\draw (110:5cm) -- (150:2.7cm);
\draw (190:5cm) -- (150:5.3cm);
\draw (190:5cm) -- (150:2.7cm);
\draw (230:5cm) -- (270:5.3cm);
\draw (230:5cm) -- (270:2.7cm);
\draw (310:5cm) -- (270:5.3cm);
\draw (310:5cm) -- (270:2.7cm);
\draw (350:5cm) -- (30:5.3cm);
\draw (350:5cm) -- (30:2.7cm);
\draw[fill] (70:5cm) circle (3pt);
\draw[fill] (110:5cm) circle (3pt);
\draw[fill] (190:5cm) circle (3pt);
\draw[fill] (230:5cm) circle (3pt);
\draw[fill] (310:5cm) circle (3pt);
\draw[fill] (350:5cm) circle (3pt);
\draw (70:5cm) node [anchor=south west] {$a_2$};
\draw (110:5cm) node [anchor=south east] {$b_1$};
\draw (190:5cm) node [anchor=east] {$a_1$};
\draw (230:5cm) node [anchor=north east] {$b_3$};
\draw (310:5cm) node [anchor=north west] {$a_3$};
\draw (350:5cm) node [anchor=west] {$b_2$};
\draw[rotate=30,fill=white] (4,0) ellipse (1.7cm and .8cm);
\draw[rotate=150,fill=white] (4,0) ellipse (1.7cm and .8cm);
\draw[rotate=270,fill=white] (4,0) ellipse (1.7cm and .8cm);
\draw (30:4cm) node {$C_2$};
\draw (150:4cm) node {$C_1$};
\draw (270:4cm) node {$C_3$};
\end{tikzpicture}
    \caption{The graph $N_k$.}
    \label{fig:graph-N_k}
\end{figure}
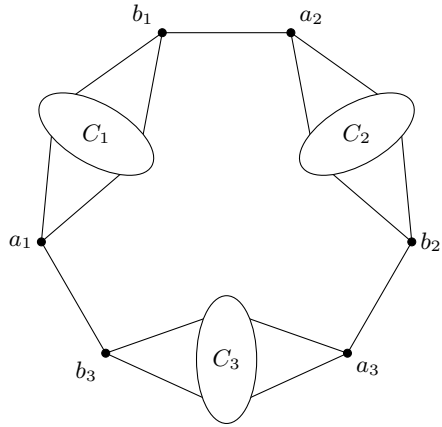

\begin{lemma}
For every $ k \geq 3$, the graph $ N_k $ is $k$-colorable.
\end{lemma}
\begin{proof}
For every $ t \in \{1,2,3\}$, color $ a_t $ and $ b_t $ with color $ t $ and color the $ k-1 $ vertices in $ C_t $ with the $ k-1 $ colors $ \{ 1,2, \dots, k \} \setminus \{t\} $. It is easy to check that this is a proper coloring of the vertex-set of $ N_k$ and as $ k \geq 3 = t $, we have used exactly the $ k $ colors $ \{1, 2, \dots, k\}$. \qed
\end{proof}


\begin{lemma} \label{lem:labels-of-Nk-are-nice}
Assume $ k\geq 3$ is an integer and set $ G = N_k $. If $ \ell: V(G) \rightarrow \{0,1\} $ is a binary labeling of $  G $, and $ S$ is the score matrix of $ (G, \ell)$, then there exists an integer $ j \in \{0,1,\dots, k+1\} $ such that $ S_{0,j} S_{1,j} \neq 0$. 
\end{lemma}
\begin{proof} 
If there exists a vertex with label $0$ whose neighbors all have label $0$ as well, then $ S_{0,0} = 1$, and since by definition, $S_{1,0} = 1 $, choosing $ j=0$ gives us the required result. Similarly, if there exists a vertex with label $1$ whose neighbors all have label $1$ as well, then $ S_{1,k+1} = 1$, and as $ S_{1,k+1} = 1$, choosing $ j = k+1$ gives us the required result. Hence, from now on, we may assume that
\begin{equation} \label{assumption1}
    \text{there exists no vertex in $ G $ that has the same label as all its neighbors.}
\end{equation}

For $t \in \{1,2,3\}$, set $ C_t = N(a_t) \cap N(b_t) $. Notice that $C_t $ induces a clique of size $ k-1 $ in $ G$. For all $ t \in \{1,2,3\}$, for all $ u, v \in C_t$, we have $ N[u] = N[v]$, thus $ \score(u) = \score(v) $. If two distinct vertices $ u $ and $ v $ of $ C_t $ have different labels, then by choosing $ j = \score(v) $, we have $S_{0,j}S_{1,j} \neq 0$. So, we may assume that for every $ t \in \{1,2,3\} $,
\begin{equation}
    \text{all the vertices of $ C_t $ have the same label.}  \label{assumption2}
\end{equation}
Thanks to~(\ref{assumption2}), for the rest of this proof and by abuse of notion, we use the term \emph{the label of $ C_t$} for referring to the common label of the vertices of $ C_t $, and we denote it by $ \ell(C_t)$.
Also, notice that by~(\ref{assumption1}), if $ \ell(C_t) = i$, where $ i \in \{0,1\} $, then at least one of the vertices $ a_t $ and $ b_t $ must receive the label $ 1-i $. Thus, for every $ t \in \{1,2,3\} $, 
\begin{equation}
    \text{at least one of $ a_t$ and $b_t$ has a label different from the one of $ C_t$.} \label{assumption3}
\end{equation}

\newcommand{\lbl}{\hat i} 

Since there are three indices $\{1,2,3\}$, but only two labels $\{0,1\}$,  there exist $ t, t' \in \{1,2,3\} $ such that $t\neq t'$ and $ \ell(a_{t}) = \ell(a_{t'}) $. By symmetry, we may assume $ \ell(a_1) = \ell(a_2) = \lbl $ for some $ \lbl  \in \{0,1\}$. 
Notice that for every $ u \in C_1$, $N[u]\setminus \{a_1\} = N[b_1] \setminus \{a_2\} $, so $ \score(u) = \score(b_1)$. Thus if the label of $ C_1$ is different from the label of $ b_1$, then choosing $j = \score(b_1)$ completes the proof. So, we may assume that $\ell(C_1) = \ell(b_1)$. Therefore, by~(\ref{assumption3}), we must have $ \ell(C_1) \neq \ell(a_1) $. Consequently, $ \ell(b_1) = \ell(C_1) = 1- \lbl $.

Now, if $  \ell(b_3) = 1-\lbl  $, then for every $ u \in C_1 $ we have $\score(a_1) = \score(u) $. And because $ \ell(C_1)\neq \ell(a_1) $, choosing $ j = \score(a_1)$ completes the proof. Hence we assume $ \ell(b_3) = \lbl$. Now, because of~(\ref{assumption1}), there must be a vertex $ u $ in the neighborhood of $ b_3 $ with label $ 1-\lbl $. As we already have $ \ell(a_3) = \lbl $, we must have $ u \in C_3$, and therefore by~(\ref{assumption2}), $\ell(C_3) = 1-\lbl$. 

Moreover, if $ \ell(C_2) = 1-\lbl $, then we have $ \score(a_2) = \score(b_1) $, and as $a_2$ and $b_1$ have different labels, choosing $ j = \score(a_2)$ completes the proof. So, we also assume that $ \ell(C_2) = \lbl $. Thus $ \ell(C_2) = \ell(a_2) = \lbl $. Therefore, by~(\ref{assumption3}), we must have $\ell(b_2) = 1 - \lbl $. Now, notice that the neighbors of $ a_3 $ all have label $ 1 - \lbl $, thus by~(\ref{assumption1}), we must have $ \ell(a_3) = \lbl $.

The labels of vertices of $ G $, with all the assumptions so far, are as shown in Figure~\ref{fig:all-labels-in-lemma}.

\begin{figure}[h]
    \centering
    \begin{tikzpicture}[scale=.5]
{\small
\draw (70:5cm) -- (110:5cm);    
\draw (190:5cm) -- (230:5cm);
\draw (310:5cm) -- (350:5cm);
\draw (70:5cm) -- (30:5.3cm);
\draw (70:5cm) -- (30:2.7cm);
\draw (110:5cm) -- (150:5.3cm);
\draw (110:5cm) -- (150:2.7cm);
\draw (190:5cm) -- (150:5.3cm);
\draw (190:5cm) -- (150:2.7cm);
\draw (230:5cm) -- (270:5.3cm);
\draw (230:5cm) -- (270:2.7cm);
\draw (310:5cm) -- (270:5.3cm);
\draw (310:5cm) -- (270:2.7cm);
\draw (350:5cm) -- (30:5.3cm);
\draw (350:5cm) -- (30:2.7cm);
\draw[fill] (70:5cm) circle (3pt);
\draw[fill] (110:5cm) circle (3pt);
\draw[fill] (190:5cm) circle (3pt);
\draw[fill] (230:5cm) circle (3pt);
\draw[fill] (310:5cm) circle (3pt);
\draw[fill] (350:5cm) circle (3pt);
\draw (80:4.8cm) node [anchor=south west] {$\ell(a_2)=\hat i$};
\draw (115:5.2cm) node [anchor=south] {$\ell(b_1)=1-\hat i$};
\draw (190:5.1cm) node [anchor=east] {$\ell(a_1)=\hat i$};
\draw (235:4.8cm) node [anchor=north east] {$\ell(b_3)=\hat i$};
\draw (305:4.8cm) node [anchor=north west] {$\ell(a_3)=\hat i$};
\draw (350:5.1cm) node [anchor=west] {$\ell(b_2)=1-\hat i$};
\draw[rotate=30,fill=white] (4,0) ellipse (1.7cm and .8cm);
\draw[rotate=150,fill=white] (4,0) ellipse (1.7cm and .8cm);
\draw[rotate=270,fill=white] (4,0) ellipse (1.7cm and .8cm);
\draw (30:4cm) node {$C_2$};
\draw (150:4cm) node {$C_1$};
\draw (270:4cm) node {$C_3$};
\draw (25:5.4cm) node [anchor=south west] {$\ell(C_2)=\hat i$};
\draw (155:5.4cm) node [anchor=south east] {$\ell(C_1)=1-\hat i$};
\draw (270:5.4cm) node [anchor=north] {$\ell(C_3)=1-\hat i$};
}
\end{tikzpicture}
    \caption{The labels at the end of the proof of Lemma~\ref{lem:labels-of-Nk-are-nice}.}
    \label{fig:all-labels-in-lemma}
\end{figure}
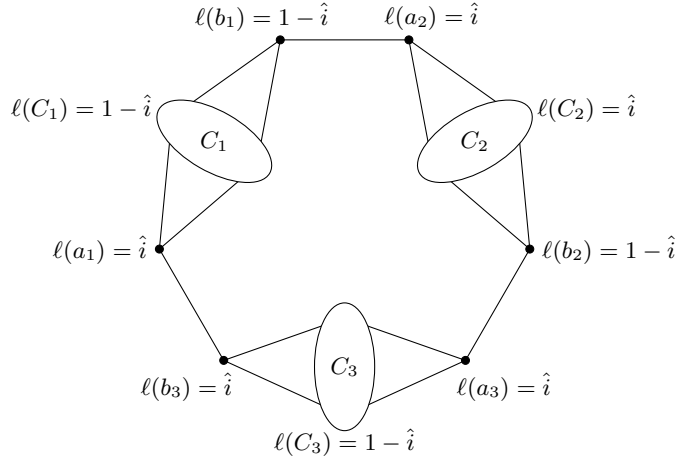

To conclude, consider $ a_1 $ and a vertex $ u \in C_3 $. In their closed neighborhoods, they both have twice the label $ \lbl $ and $ k-1 $ times the label  $1- \lbl $, so $\score(a_3) = \score(u)$. Moreover, they have different labels. So, by choosing $ j = \score(a_3)$ we have the required result. \qed
\end{proof}

\subsection{Anonymous case}

\begin{theorem}
For every $ k \geq 3 $, $k$-colorability is not certifiable by binary certificates in the anonymous model.
\end{theorem}
\begin{proof}
Let $ G $ be the graph $ N_k $ and $ H $ be a complete graph on $ k+1 $ vertices. 
Let $ \ell $ be a binary labeling of $ G $, and let $S $ be the score matrix of $ (G, \ell)$. Since $ G $ is $k$-colorable and $ H $ is not, by Lemma~\ref{lem:techniques}, to prove the theorem, it is enough to find a binary labeling $\ell':V(H) \rightarrow \{0,1\} $ of $ H $  such that every view in $ (H, \ell') $ is a view in $ (G,\ell) $.

By Lemma~\ref{lem:labels-of-Nk-are-nice}, there exists $ j \in \{0,1,\dots, k+1\} $ such that $ S_{0,j} $ and $S_{1,j} $ are non-zero, 
meaning that:
\begin{itemize}
\item
if $ j = 0$, then there exists a vertex $ u \in V(G) $ such that $\ell(u) = 0 $ and all its neighbors have label $0$,
\item
if $ j = k+1 $, then there exists a vertex $ v \in V(G) $ such that $\ell(v) = 1 $ and all its neighbors have label $1$, or,
\item
if $ 0 < j < k+1 $, then there exist two distinct vertices $ u,v \in V(G) $ such that $ \ell(u)=0$, $ \ell(v) = 1$, and there are $j$ vertices of label $1$ in $N(u)$, and $ j-1$ vertices of label $1$ in $N(v)$.
\end{itemize}

Let $ V(H) = U \cup V $ such that $ |U| = k+1 -j $ and $ |V| = j $. If $j \in \{0, k+1\}$, then $U$ or $V$ is an empty set. Notice that the condition on the cardinalities of $ U $ and $ V $ implies that $ U \cap V = \varnothing $. Define:
$$
\ell'(w) = 
\begin{cases}
0 & w \in U \\
1 & w \in V
\end{cases}
$$
Notice that if $ j = 0 $ or $ j = k+1 $, then $ \ell'(\cdot) = 0 $ or $ \ell'(\cdot) = 1 $ respectively. 

The view of each vertex of $ U $ in $ H $ is the same as the view of $ u $ in $ G $ and the view of each vertex of $ V $ in $ H $ is the same as the view of $ v $ in $ G $. So, every view in $ H $ is the same as some view in $ G $. \qed
\end{proof}

\subsection{Extension to identifiers}

\begin{theorem}
For every $ k \geq 3 $, $k$-colorability is not certifiable by binary certificates in the proof-labeling scheme model when the range of the identifiers for a graph on $ n $ vertices is $ f(n) = n^3+3n $.
\end{theorem}
\begin{proof}
Set $K=(k+1)^2+1$ and $ \Lambda = \{1, 2, \dots, K\}$.

Let $ G_1, G_2, \dots, G_K$ be $K$ connected graphs each isomorphic to~$N_k$. Notice that $ |V(G_i)| = 3k+3 $. For $ i \in \Lambda $, consider an identifier $ \id_i :~V(G_i) \rightarrow [1, f(|V(G_i)|)] $ of $ G_i $ such that the vertices of $ G_i $ receive ID's in the range $[(i-1)(3k+3)+1, i(3k+3)]$. Notice that 
\begin{equation*}
(i-1)(3k+3)+1 \geq (1-1)(3k+3)+1 = 1,
\end{equation*}
and 
\begin{equation*}
\begin{split}
    i(3k+3) & \leq K(3k+3) \\ 
    & = \big( (k+1)^2+1 \big) (3k+3) \\
    & = \frac{(3k+3)^3}{27}+(3k+3) \\
    & \leq (3k+3)^3 + 3(3k+3) \\
    & = f(3k+3) = f(|V(G_i)|).
\end{split}    
\end{equation*}
So, $[(i-1)(3k+3)+1, i(3k+3)] \subseteq [1, f(|V(G_i)|)]$. Thus the identifiers $\id_i$ exist.

Let $ H $ be a complete graph on $ k+1 $ vertices. Notice that each $ G_i $ is $k$-colorable and $H$ is not. Hence, by Lemma~\ref{lem:techniques}, to prove the theorem, it is enough to define an identifier $ \id: V(H) \rightarrow [1, f(|V(H)|)] $ and a binary labeling assignment $ \ell:V(H) \rightarrow \{0,1\}$ of $ H $ such that every view in $(H, \id, \ell)$ is the same as a view in $ (G_i, \id_i, \ell_i) $ for some $ i\in \Lambda$. 

Let $ S^{(i)} $ be the score matrix of  $(G_i, \ell_i)$. By Lemma~\ref{lem:labels-of-Nk-are-nice}, for every $ i \in \Lambda $, there exists $ j \in \{0,1,\dots, k+1\}$ such that $ S^{(i)}_{0,j}S^{(i)}_{1,j} \neq 0$. Therefore, by the pigeonhole principle, there exists an integer $ j \in \{0, 1, \dots, k+1\}$ and a subset $ \Lambda_0 $ of $ \Lambda $ such that $ |\Lambda_0| \geq k+1 $ and for all $i \in \Lambda_0$, we have $ S^{(i)}_{0,j}S^{(i)}_{1,j} \neq 0 $.

Notice that $j \in \{ 0,1, \dots, k+1\} $, hence 
$$ |\Lambda_0| \geq k+1 \geq k+1 -j  \text{ and } |\Lambda_0| \geq k+1 \geq j. $$
Thus, we can find $ j $ distinct vertices $ v_1, v_2, \dots, v_j \in \bigcup_{i \in \Lambda_0} V(G_i) $ with label 1 and score $j$, and $ k+1 - j $ distinct vertices $ v_{j+1}, v_{j+2}, \dots, v_{k+1} \in \bigcup_{i \in \Lambda_0} V(G_i) $ with label 0 and score $j $. Notice that in case that $ j = 0 $ (resp.\ $ j = k+1$), then the first (resp.\ the second) set of vertices in empty.

Now, assume $ V(H) = \{ u_1, u_2, \dots, u_{k+1} \} $. 

First, define an identifier $ \id: V(H) \rightarrow [1, f(|V(H)|)] $ as follows: 
$$ \id(u_t) = \id(v_t) \text{ for every }  t \in \{1,2,\dots, k+1\}.$$ 
Notice that by this definition, for all $ t$, we have:
$$\id(u_t) \in \bigcup_{i \in \Lambda} [(i-1)(3k+3)+1, i(3k+3)] = [1, K(3k+3)]. $$
On the other hand: 
\begin{equation*}
\begin{split}
K(3k+3) = & \frac{(3k+3)^3}{27}+(3k+3) \\
& = (k+1)^3 + 3(k+1) \\
& =  f(k+1) = f(|V(H)|),
\end{split}
\end{equation*}
and therefore the image of $\id $ is a subset of $[1, f(|V(H)|)]$. 

Second, define a binary labeling $ \ell:V(H) \rightarrow \{0,1\} $ as follows: 

$$\ell(u_t) = \begin{cases}
1 & t \leq j \\
0 & j+1 \leq t
\end{cases}.$$

Notice that if $ j = 0 $ or $ j = k+1 $, then $ \ell'(\cdot) = 0 $ or $ \ell'(\cdot) = 1 $ respectively.

Thus, for every $ t $, $1 \leq t \leq k+1 $, the view of vertex $u_t$ in $(H, \id, \ell)$ is exactly the view of vertex $ v_t $ some of the $G_i$'s, $ i \in \Lambda_0$. \qed
\end{proof}


\section{Going below $\lceil \log k \rceil$}

In this section, we illustrate that in some cases one can go below the natural upper bound. 
More precisely, we exhibit an LCL whose natural encoding uses $k$ different labels, and for which one can find a certification that a solution exists with strictly less than $k$ labels.
The problem we use is \emph{distance-2 3-coloring}, which is the same as 3-coloring except that even having two nodes at distance 2 colored the same is forbidden. The natural encoding consists in giving the colors (that are between $1$ and $3$) to the vertices, but we show that one can actually certify distance-2 3-colorability with just two different certificates.

\begin{lemma}\label{lem:structure_3col}
A connected graph $G$ is distance-2 3-colorable if and only if it is a a cycle of length $0\!\mod 3$ or a path. 
\end{lemma}
\begin{proof}
First, let $G$ be a connected distance-2 3-colorable graph.
Each vertex and its neighbors form a set of vertices that are pairwise at distance at most 2, so they should all have different colors. Since there are only three available colors, $G$ has maximum degree at most $2$. 
Thus, $G$ is a path or a cycle.

If $G$ is a path, then we are done, so assume $G$ is a cycle. 
Assume that the cycle is $v_1,v_2,\dots,v_k,v_1$. Consider a proper distance-2 3-coloring of $G$ with colors $\{0,1,2\}$. Notice that if a vertex $v$ of $ G $ has color $i \mod 3$, then necessarily, its two neighbors have colors $(i-1) \mod 3$ and $(i+1) \mod 3$. Without loss of generality, assume that the color of $ v_1 $ is 1 and the color of $ v_2$ is 2. Let $ [t]$ denote the remainder of $ t $ divided by $ 3$. So, $[t] \in \{0,1,2\}$. We prove by induction on $ t $ that $ v_t $ has color $ [t] $ for every $ t \in \{1,2,\dots,k\}$. This holds by assumption for $ t=1,2$. Now let $ t \geq 3$ and assume that $ v_{t'}$ has color $ [t'] $ for every $ t' < t$. By the induction hypothesis, the color of $ v_{t-1}$ is $ [t-1]$, hence the colors of its two neighbors, namely $v_{t-2}$ and $ v_t$, are $[t]$ and $[t+1]$. Again, by induction hypothesis, the color of $ v_{t-2}$ is $ [t-2] = [t+1]$. Thus the color of $ v_t$ must be $[t]$, proving the statement. Therefore, the color of $ v_k $ is $ [k]$. Finally, notice that because $ v_1 $ has color 1, and $ v_2$ has color $ 2 $, the color of $ v_k$ must be 0. Therefore $ [k] = 0$, meaning that the length of the cycle, $ k $ is equal to $0 \mod 3$.

Second, in both cycles of length $0 \mod 3$ and in paths, it is possible to find a proper distance-2 3-coloring, by simply choosing the color of one vertex and propagating the constraints.
\qed
\end{proof}

\begin{theorem}
One can certify distance-2 3-colorable graphs with a binary certification.
\end{theorem}

\begin{proof}
We prove that paths and cycles of length $0\!\mod 3$ can be recognized with a binary certification, by Lemma~\ref{lem:structure_3col}, this will give the desired result.

The idea of the certification is to give certificates of the form: 
\[...1,0,0,1,0,0,1,0,0...\]
Let us describe first the verifier algorithm of a vertex $v$: 
\begin{enumerate}
    \item If the degree is strictly more than 2, reject,  
    \item If the degree is 2, accept if and only if, $\score(v)=1$, that is, among $v$ and its neighbors, exactly one has a 1 as certificate.
    \item If the degree is 1, accept. 
\end{enumerate}
Because of Item 1, only paths and cycles can be accepted. For paths, the following labeling makes every vertex accept: choose one endpoint $u$, and give label 1 to a vertex $v$, if and only if, the distance from $u$ to $v$ is $0 \mod 3$.
For cycles of length $0 \mod 3$, the following labeling makes every vertex accept: choose an orientation and an arbitrary vertex $u$, and again give label 1 if and only the distance from $u$ to $v$ in the sense of the orientation is $0 \mod 3$. 
It is easy to see that in cycles of length different from $0 \mod 3$, at least one vertex will reject with Item 2.

Hence, we have a proper certification that a graph is a path or a cycle of length $0 \mod 3$, with only two different labels (that is, just one bit).
\qed
\end{proof}

\section{Challenges and open questions}

We have proved the one bit is not enough for certifying $k$-colorability for $k \geq 3$.  
We conjecture that the answer to Question~\ref{question:logk} is negative, that is, that $\lceil \log k \rceil$ is the optimal certification size for $k$-colorability, and this is the very first step. 

There are several challenges to overcome before one can hope to prove the conjecture. First, it would be nice to have a more general model, where the nodes can see their neighbors' identifiers, or at least distinguish them, and even better, where the nodes can see at a larger distance. At least the first step in this direction might work by using some Ramsey argument, but then losing the upper bound on the identifier range. Second, it might be necessary to use graphs of large chromatic number that do not have large cliques as subgraph, and these have complicated structures.

A different direction is to understand other LCLs, and to try to see which ones have a certification that is more efficient than the natural encoding (such as distance-2 3-coloring) and which do not. Maybe in this direction, one could characterize exactly which properties can be certified with one bit.

\bibliographystyle{plain}
\bibliography{biblio-JGA}
\end{document}